\documentclass[10pt,leqno]{amsart}
\usepackage{graphicx}
\usepackage{indentfirst,csquotes}

\usepackage{graphicx}
\usepackage{amsmath,amsthm,amssymb}
\usepackage{graphicx}
\usepackage{tabularx}
\usepackage[colorlinks=true, allcolors=blue]{hyperref}
\usepackage{tcolorbox}
\usepackage[T1]{fontenc}
\usepackage{listings}
\usepackage{xcolor}
\usepackage{physics}
\usepackage{subcaption}   
\usepackage{algorithm}
\usepackage[noend]{algpseudocode}
\usepackage[noend]{algpseudocode}
\usepackage{comment}

\newtheorem{example}{Example}[]
\newtheorem{theorem}{Theorem}[]
\newtheorem{proposition}{Proposition}[]

\begin{document}

\title{From Physical to Logical: Graph-State-Based Connectivity in Quantum Networks}

\author{Mateo M. Blanco,  Manuel Fernández-Veiga, Ana Fernández-Vilas, Rebeca P. Díaz-Redondo}
\email{mveiga@det.uvigo.es avilas@uvigo.es rebeca@det.uvigo.es}

\let\thefootnote\relax
\footnotetext{This work was supported under the grants TED–2021–130369B–C31 funded by MICIU/AEI/10.13039/501100011033 by the “European Union NextGenerationEU/PRTR”;  PID2023-148716OB-C31 funded by MICIU/AEI/10.13039/501100011033;  grant ED431B 2024/41 (GPC) by the Galician Regional Government; and the regional agreement for laboratories and demonstration centers in cybersecurity within the RETECH program.}

\address{atlanTTic, Universidade de Vigo, Spain}

\maketitle

\begin{abstract}
Entanglement is a key resource in quantum communication, but bipartite schemes are often insufficient for advanced protocols like quantum secret sharing or distributed computing. Graph states offer a flexible way to represent and manage multipartite entanglement in quantum networks, enabling logical connectivity through local operations and classical communication (LOCC). In this work, we extend existing approaches based on bi-star configurations to more complex multi-star topologies. We analyze the maximum connectivity that can be achieved in networks of $m$ switches, each connected to $n$ clients, including asymmetric cases where the number of clients varies per switch. We also propose methods to enable logical communication between distant nodes. Our results support the development of scalable quantum networks with rich connectivity beyond traditional bipartite structures. 

\end{abstract}

\section{Introduction}
\label{sec:intro}

Entanglement and, notably, the distribution of entangled pairs along with 
entanglement swapping and purification forms the foundations 
in the development of the quantum internet and quantum communications, 
with its unique properties enabling secure and efficient information 
transfer across distant locations~\cite{Cacciapuoti2020, Shi2024, 
Valivanti2020, Wallnofer2019}.  However, bipartite entanglement becomes 
insufficient when exploring complex quantum  protocols beyond simple QKD, 
including quantum secret sharing, quantum error correction, and  multipartite QKD.
More complex entangled states, such as GHZ states,  are challenging to
generate using basic entanglement processes such as entanglement swapping.  
To address these shortcomings, Graph States~\cite{Hein2006} offer a structured 
framework for generating multipartite entanglement in quantum networks by 
associating entanglement with graph edges and enabling transformations via 
LOCC operations~\cite{Vandennest2004}. 

This approach allows for flexible, scalable entanglement distribution while 
preserving the network’s physical structure. Theoretical~\cite{Guhne2005, 
Markham2008, Vandre2024} and experimental~\cite{Bell2014, Huang2023} 
studies confirm their benefits in optimizing network resources. From the quantum 
networking perspective, graph states have been proposed as the core architecture 
for all-photonic quantum repeaters within a network~\cite{Benchasattabuse2024}. 
Applying LOCC operations and measurements to graph states alters the topology of 
the logical network and allows one to connect initially disconnected nodes. As 
classical network reachability supports the ability of a node to successfully
communicate with another node in the network, quantum reachability should establish 
and maintain entanglement between two nodes anywhere in the network, enabling
communication and, as an extension, distributed quantum protocols.

In this paper, we follow the approach in~\cite{Chen2024} to obtain artificial 
topologies from the physical ones by manipulating bi-colorable graph states 
through local operations. Reference \cite{Chen2024} provides protocols to achieve various
connectivity schemes  from a bi-star configuration: peer-to-peer---all leaves of 
one star graph are connected to a leaf in the other; role delegation---one of the 
leaves acts as the center of the star topology; and extranet---all leaves are
interconnected, forming a $\ket{K_{n_1,n_2}}$ graph state. In parallel, Quantum
Local Area Networks (QLANs) are introduced in~\cite{Mazza2024}, as a star graph 
with a central switch and multiple leaf nodes representing clients. When central
nodes in the stars are connected by quantum channels, a bi-star configuration 
emerges, i.e., two star graphs joined through entanglement of their respective 
centers. Finally,~\cite{deJong2024} establishes a strong upper bound on the maximum
neighbor connectivity in a linear cluster state, $\alpha \leq \lfloor 
\frac{n+3}{2} \rfloor$, where $n$ is the number of elements in the cluster.

We seek the generalization of the approaches above to extend their applicability 
to more complex network architectures. Specifically, for a  linear graph consisting 
of $m$ stars, each with $n$ nodes (leaves), the maximum entangled states that can 
be generated have not been fully characterized. In this work, we partially 
characterize these maximal entangled states, leveraging multipartite entanglement 
and its implementation in graph states to overcome the limitations of traditional 
bipartite entanglement-based communication. Precisely, our contributions are as follows:  (1) characterizing the maximum achievable connectivity for a general multi-star network with $m$ switches, each connected to $n$ nodes (leaves);   (2) analyzing scenarios where each switch has a different number of nodes; and   (3) developing a method to enable communication between two or more nodes in the network.

\section{Theoretical preliminaries}
\label{sec:background}

Since the formalism of graph states builds upon the basic concepts of graph theory,
we briefly review in this Section the necessary definitions and
properties.

A graph $G = (V, E)$ consists of a set of vertices $V$ connected pairwise by edges $E$. 
Its extension to graph states in quantum mechanics is both straightforward and profound: 
the vertices of the graph correspond to qubits, each initialized in the state 
$\ket{+} = \frac{1}{\sqrt{2}}(\ket{0}+\ket{1})$, while the edges represent 
controlled-$Z$ (CZ) gates that entangle the connected qubits. Under these assumptions, 
a graph state is formally defined as  
\begin{equation*}
    \ket{G}=  \prod_{e \in E} CZ \ket{+}^{\otimes \abs{V}}.
\end{equation*}
where CZ, in the basis $\{ \ket{00}, \ket{01}, \ket{10}, \ket{11}\}$, applies a Pauli-$Z$ 
to the second input qubit  only when the first qubit is $\ket{1}$. A simple example
illustrates the entanglement  structure: a graph consisting only of two connected 
vertices forms a Bell state, i.e.,
$\ket{G} = CZ \left(\ket{+}, \ket{+} \right) = \frac{1}{\sqrt{2}} (\ket{\phi^-} + \ket{\psi^+})$, where the usual Bell states  have been denoted as
$\ket{\phi^{\pm}}= \frac{1}{\sqrt{2}}(\ket{00}\pm \ket{11})$, $\ket{\psi^{\pm}} = \frac{1}{\sqrt{2}} (\ket{01} \pm \ket{10})$. Since graph states are built from their 
classical graph counterparts, their fundamental properties and basic definitions 
remain similar. First, the neighborhood of a vertex $i \in V$ in the graph $G=(V,E)$ 
is the set of nodes directly connected to $i$, $N_i = \{j \in V : (i,j) \in E \}$. 
A path is a sequence of edges that connects a sequence of distinct vertices. 
If the initial and final vertices coincide, the path 
forms a cycle. A tree is a connected acyclic graph ---a unique path 
exists between any pair of nodes---, and the nodes of degree $1$ in a tree (nodes with 
only one neighbor) are called the leaves of the tree. A caterpillar is 
a tree such that the removal of leaves results in a path. A star is a tree where
all the nodes but one are leaves, equivalently, if deletion of all the leaves results
in a single node without edges.

The logical structure of a graph state may be modified by means of the application of 
Pauli measurements. Each of these will induce local changes in the measured qubit and 
its vicinity, altering the entanglement relationships among the existing qubits and 
thus changing the structure of the graph state and correspondingly the virtual topology 
of the network of states. More specifically, these effects will be a combination of 
the following elementary transformations on a graph.  Given a subset of nodes 
$A \subset V$, the induced subgraph of $A$ will be  denoted as $G[A] = (A, E_A)$, and 
is the subgraph with edges $E_A = \{(i,j) \in E: i \in A \wedge j \in A\}$ obtained 
after deleting from  $G$ all the vertices not in $A$ and their edges.

\begin{enumerate}
\item[(i)]  \textbf{Graph complementation}. Let $G=(V,E)$ be a graph. 
The complement of $G$, denoted as $\tau(G)$ is the graph $(V, \Bar{E})$, where  
\begin{equation*}    
    \Bar{E} = V^2 \setminus E = \{(i,j) \in V ^2 : (i,j) \notin E\}.
\end{equation*}
Clearly, $G \cup \tau(G)$ is the \emph{complete graph} on $|V|$ vertices, i.e., the graph
wherein any pair of vertices is connected by an edge. The complete graph with $n$ 
nodes will be denoted hereafter as $K_n$.

       \item[(ii)] \textbf{Local complementation}. Of particular interest for graph states is the 
local complementation at a vertex $i$.  This operation consists of taking the complement 
of the neighborhood of $i$. Accordingly, the local complementation at $i$ is
\begin{equation*}
    \tau_i(G) = \bigl( V, (E \cup N_i^2) \setminus E_{N_i} \bigr),
\end{equation*}
where $E_{N_i}$ is the set of edges in the neighborhood $N_i$ of $i$. The resulting graph
$\tau_i(G)$ has the same vertex set $V$ as $G$, but the edges within $N_i$  are deleted, 
and all the edges missing between pairs of nodes in $N_i$
are appended. This local operation corresponds to a sequence of local Clifford operations, 
as described in~\cite{Vandennest2004}.

       \item[(iii)]  \textbf{Vertex deletion}, where a graph $G$ is transformed into $G_{-i}$ by 
removing a vertex $i$ and all its associated edges:
\begin{equation*}
    G_{-i} = \bigl( V\setminus \{ i \}, E\setminus (\{i\} \times N_i) \bigr).
\end{equation*}
\end{enumerate}

Though the operation of local complementation may initially appear abstract, it actually 
aligns with a fundamental concept in quantum mechanics: \textbf{Local Unitary equivalence (LU)}. 
Two $n$-qubit graph states, $\ket{G}$ and $\ket{G'}$, are said to be LU equivalent 
if there exists a set of local unitaries $\{ U_i \}_{i\in I}$ such that  
\begin{equation*}
    \ket{G'} = \bigotimes_{i \in I} U_i \ket{G},
\end{equation*}
so that it is possible to transform $\ket{G}$ into $\ket{G^\prime}$ by a sequence of local
quantum operations. For the special case of local complementation of a vertex $a$ and its
neighborhood $N_a$, the corresponding unitary transformation is given by  
\begin{equation}
\label{unitary} 
    U_a^T = e^{-\imath \frac{\pi}{4} \sigma_{X_a}} \bigotimes_{b \in N_a} e^{\imath 
    \frac{\pi}{4}  \sigma_{Z_b}}, 
\end{equation}
where $\sigma_{X_a}$ is the application of a Pauli $X$ operator to vertex $a$.  Indeed, 
local complementation is an LU equivalence, with the corresponding unitary transformation 
given by~\autoref{unitary}.

\begin{figure}[t]
    \centering
    \includegraphics[width=\linewidth]{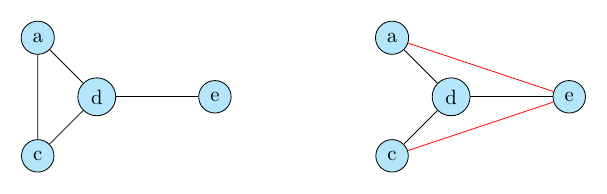}
    \caption{Example of local complementation at vertex $d$. Added edges are colored in red.}
    \label{fig:LC}
\end{figure}

\begin{example}[Local complementation]
\label{example:startoKn}  As shown in the ~\autoref{fig:LC}, when vertex $d$ is selected, the edges connecting its neighbors that were previously present are removed, while new edges are formed between those neighbors that were not previously connected. 

On the other hand, an useful example of local complementation and LU equivalence is the 
    transformation between a complete graph $K_n$ of $n$ vertices $\ket{K_n}$ and a
    star graph with $n - 1$ leaf nodes and a central node, $\ket{S_{n-1}}$. This is
    achieved by performing local complementation of $K_n$ at any node, so that this node becomes the center of $S_{n-1}$. 
\end{example}

The effect of Pauli measurement on a graph state can be summarized in the following
key result~\cite{Hein2006}
\begin{theorem}[Pauli Measurements on Vertex \( i \)]
\label{thm:Pauli-graph-measurements}
    Let $G = (V,E)$ be a graph with $\abs{V} = n$ nodes. The graph $\Tilde{G} = (\Tilde{V}, \Tilde{E})$, with $|\tilde{V}|=  n- 1$, obtained after a Pauli measurement on
    qubit $i$ is as follows.
    \begin{enumerate}
        \item[(i)] Pauli $Z$. $\sigma_{Z_i}(G) := \Tilde{G} = G_{-i}$  the deletion of 
        vertex $i$.
        \item[(ii)] Pauli $Y$. $\sigma_{Y_i}(G) := \tilde{G} = (\tau_i(G))_{-i}$ 
        corresponds to a local complementation at $i$, followed by the deletion of $i$.
        \item[(iii)] Pauli $X$. $\sigma_{X_i}(G) := \tilde{G} = \tau_{k_0}(\tau_i(\tau_{k_0}(G))_{-i})$ for an arbitrary node $k_0 \in N_i$. The new 
        graph is the local  complementation of $G$ at the neighbor  $k_0$, followed by a 
        local complementation at $i$, the deletion of $i$, and a final local complementation 
        at $k_0$.
    \end{enumerate}
\end{theorem}

\section{Problem Statement}
\label{sec:problem-statement}

By using ~\autoref{thm:Pauli-graph-measurements}, it should be possible to transform 
an initial graph state $\ket{G}$ into another one $\ket{G^\prime}$ particularly suited 
to some specific quantum protocol, or to manipulate an initial multipartite quantum 
state for some communication task (quantum or entanglement-assisted but classical).
The sequence of measurements that lead to such a desired target state will be termed 
a \emph{graph protocol}. However, since $\operatorname{CZ}^2 = I$, we need that 
the multipartite states that our graph protocol manipulates must be bi-colorable, 
where a graph is said to be bi-colorable  when  its vertex set $V$ can 
be partitioned into two subsets $\{P_1, P_2\}$ such that no two vertices within the 
same subset are adjacent. Many interesting quantum network models naturally emerge 
as bi-colored graphs, allowing for the  derivation of  practical results. 
Under this assumption, suppose $\ket{G}$ is
a graph state whose underlying graph $G$ is bi-colorable. We seek to characterize
what bi-colorable graph states $\ket{G^\prime}$ are achievable through a graph 
protocol $\mathcal{T}$ composed of a sequence of Pauli measurements, $\ket{G} \overset{\mathcal{T}}{\longrightarrow} \ket{G^\prime}$. The 
graph protocol  must avoid those sequences that may disrupt bi-colorability. 

\begin{figure*}[th!]
    \centering
    \includegraphics[width=\textwidth]{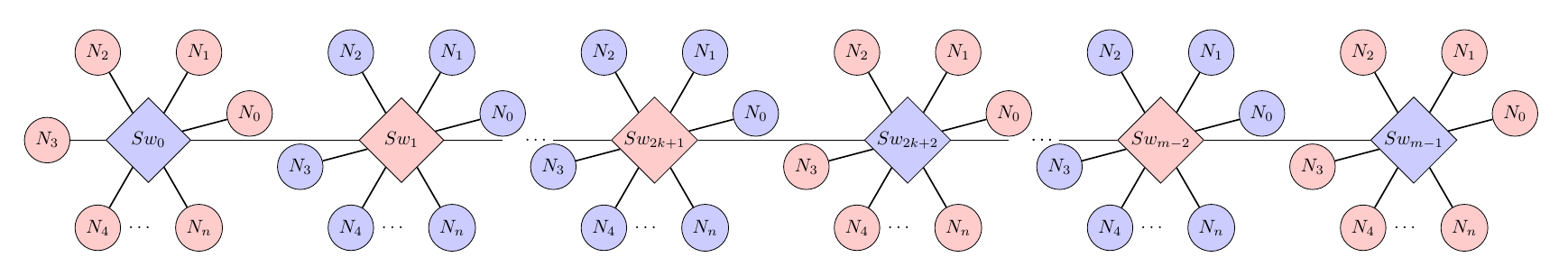}
    \caption{Multi-star configuration as an example of a linear network.}
    \label{fig:multistar}
\end{figure*}

Bi-colorable graphs are of significant importance in quantum graph states theory 
because they enable the implementation of error correction and purification
protocols~\cite{Aschauer2005, Hein2004}. This idea has been further explored in 
the context of so-called quantum local area network (QLAN) 
architectures~\cite{Mazza2024b}. Despite the constraints imposed by the
bi-colorable condition, the advantages provided by this class of graphs outweigh these 
restrictions. One such canonical graph is the binary star graph:
two star graphs joined at their central nodes. It will be denoted as
$S_{n_1,n_2}$, where $Sw_{1}$ represents the center of the first star and $Sw_{2}$ 
is the center of the second star. The remaining vertices are partitioned into 
sets $N_i^1$ and $N_j^2$, where $i \in \{1, \dots, n_1\}$ and $j \in \{1, \dots, n_2\}$. 
This binary star graph can be constructed from two independent star graphs via a 
remote CZ operation applied to their central nodes. The extension of such 
entanglement generation between their centers converts an arbitrary number of 
isolated stars connected into a connected multi-star configuration, 
see~\autoref{fig:multistar}. 

\section{Achievable Multiparty Connectivity Through Graph-state Manipulations}
\label{sec:main-results}

Consider a network in the form of a multi-star topology, as shown 
in~\autoref{fig:multistar}. The multi-star is a sequence of star subgraphs where 
the centers of the stars are connected in a line graph. We will refer to the 
centers $Sw_i$ node in a star $i$ as a \emph{switch} and to its peripheral leaf 
nodes $N_i$ as its \emph{clients}, and assume that the  stars have a variable 
number of leaves. Our goal is to characterize the achievable network connectivity 
on these nodes, where connectivity means the establishment of multipartite 
entanglement among an arbitrary subset $S \subseteq \cup_{i \in I} N_i$ 
of clients, if this is feasible. Specifically, we assume that each star 
subgraph represents  the existence of a shared GHZ state
$\ket{\text{GHZ}_{n + 1}} := \frac{1}{\sqrt{2}} \bigl( \ket{0}^{\otimes (n + 1)} + \ket{1}^{\otimes (n + 1)} \bigr)$
among the switch and its $n$ leaves, and we seek to determine the largest GHZ state 
that can be realized in a general network consisting of $m$ of these switches using 
the minimum number of measurements on the original graph states. In summary, 
the multi-star representation is particularly useful because it aligns naturally  
with a physical quantum switch configuration with the end-nodes as the clients; 
and the extended (largest) entangled state can then be used for implementing a multi-
partite quantum protocol (see \autoref{aplications:protocols})

As a first step, we first establish a method to achieve the maximum neighbor  
configuration starting from a multi-star topology.

\subsection{Maximum neighbor configuration} 
\label{subsec:maxneigh}

Recent works~\cite{deJong2024} have explored \emph{linear} graph states (or cluster states)
and their transformation into maximal $\ket{\text{GHZ}}$ configurations. In fact, 
\cite{deJong2024} establishes a strict upper bound for an arbitrarily long linear 
cluster state, given by $\alpha_0 \leq \left\lfloor \frac{m+3}{2} \right\rfloor$,
where $m$ represents the number of linear nodes in the cluster, and the equality holds for 
an odd number of switches. While the proofs of these theorems are compelling and neat, the 
addition of leaf nodes to each switch, as in our case, further complicates the protocol 
and requires a refinement in the arguments. An important remark is that, in both cases, 
the post-measurement state is locally Clifford equivalent to the GHZ state via local
complementation. We present the modified 
measurement protocol as~\autoref{algomaxconnected}. For our problem, we have the following simple result.
\begin{proposition}
\label{alpha}
    Given a multi-star graph $G$ of $m$ switches and $n$ nodes each, $m$ is an odd number 
    of and $n\in \mathbb{N}$, the complete graph $K_{\alpha}$ is achievable with \autoref{algomaxconnected} for $\alpha = (n+1) \frac{m+1}{2}$.
\end{proposition}

\begin{proof}
    The result can be derived by following directly the steps in the algorithm. 
    We start with all the $m (n+1)$ nodes and then subtract the number of measurements 
    performed. First, we measure the leaf nodes of every other switch, totaling
    $n \frac{m-1}{2}$. After that, we apply $\sigma_X$ gates on the remaining switches, transforming the network into a star graph. This yields a total of 
    \begin{equation*}    
     m(n+1) - n\frac{m-1}{2} - \frac{m-1}{2} = (n+1)\bigl( m - \frac{m-1}{2} \bigr).
    \end{equation*}
    Thus, we arrive at the expression $\alpha = (n+1)\frac{m+1}{2}$ clients completely 
    connected. The cost of this transformation,the number of necessary measurements, is 
    $(n+1)\frac{m-1}{2}$.
\end{proof}

\begin{algorithm*}[t]
\caption{Maximally Connected Graph from linear graph state $\ket{G}$ with  odd number of $m$ switches and $n$ leaves per $S_w$.}\label{algomaxconnected}
\begin{algorithmic}[1]
\Function{AchieveMaxConnectedGraph}{$\ket{G}, m, n$}
    \State $P_0 \gets \{\text{nodes color } 0 \}$ ;  $P_1 \gets  \{\text{nodes color } 1 \}$ \Comment{bi-colored set of nodes}
    \State $E \gets E_S \cup E_{n_{i} }$ \Comment{edges between   $S_{w_i}$ $S_{w_{i+1}}$ and edges corresponding to leaves $K_j^i$ in a $S_{w_i}$}
    \State $G= (P_0, P_1, E)$ \Comment{Resulting bi-colored graph $G$}

\For{$odd \in \{m \bmod 2\}$ } 
\Comment{Leave $Sw_0$ and $Sw_{m-1}$ untouched and select Odd switches}
    \State $\sigma_Z(\{K_{odd}^k\})$ \Comment{$n$ measurements $Z$ on leaf nodes of every switch $S_{w_{odd}}$}
    \State $\sigma_X(\{S_{w_{odd}}\})$ \Comment{A measurement $X$ on every switch $S_{w_{odd}}$ with $k_0$ as its right neighbouring switch}
    \EndFor
    \State \Return $G'=({P'}_0, {P'}_1,E')$ \Comment{Maximally Connected Graph with $(n+1)\frac{(m+1)}{2}$ vertices.}
\EndFunction
\end{algorithmic}
\end{algorithm*}

\begin{figure*}[ht!]
      \centering
      \begin{tabular}{c|c}
       \begin{minipage}{0.44\textwidth} 
        \centering
        \includegraphics[width=\linewidth]{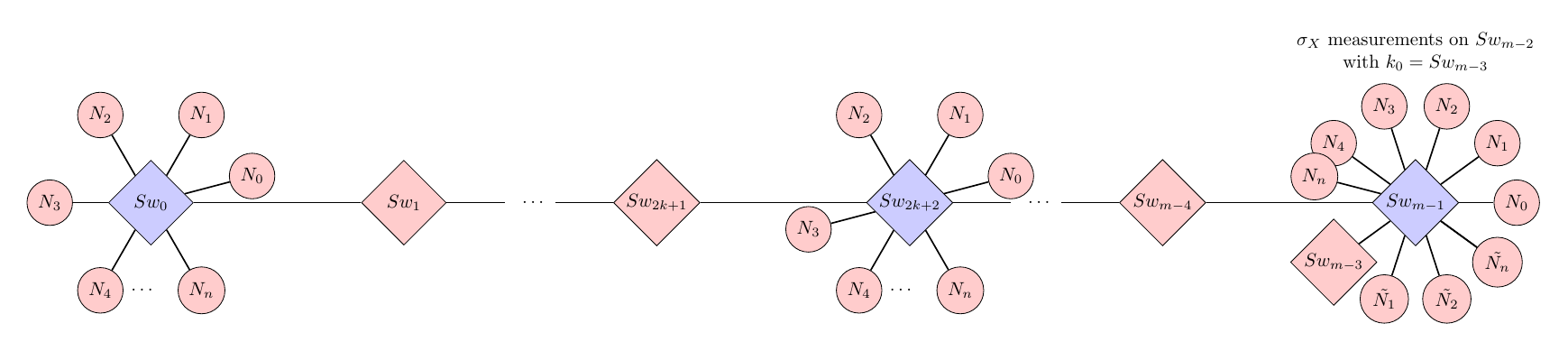}
        \label{fig:subfig3}
    \end{minipage} \hfill
           &  
             \begin{minipage}{0.5\textwidth} 
        \centering
        \includegraphics[width=\linewidth]{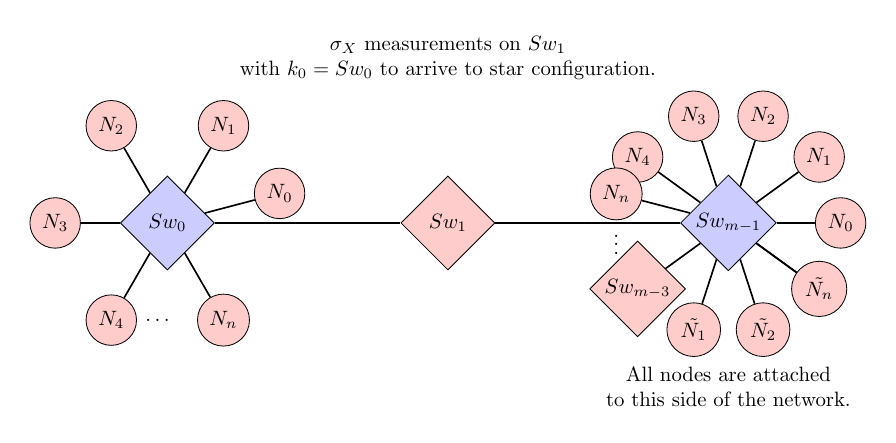}
        \label{fig:subfig4}
    \end{minipage} \\
            (a) & (b) \\
      \end{tabular}
    \caption{Steps of \autoref{algomaxconnected} for the case of an odd number of switches  over a full linear network  including the switches $Sw_i$ and their clients $N_j$.  Clients are eliminated through $\sigma_Z$ measurements in every other node. 
    (a) Starting from the second-to-last node, $\sigma_X$ measurements are performed on the lone switches, with $k_0$ being the right neighbor. 
    (b) In the final step, all nodes accumulate at the last switch. A final $\sigma_X$ measurement in $Sw_1$ results in the formation of a 
    star configuration that is LC equivalent to a complete graph.
    \label{fig:linear-network}}
\end{figure*}

Clearly, we observe that the result for $\alpha_0$ in cluster graphs, formally proven in~\cite{deJong2024} is not recovered when taking the limit $n \to 0$. This discrepancy arises from the topological differences between the two networks: while a linear cluster resembles a one-dimensional structure, a multi-star graph does not.  

Next, consider an even-numbered multi-star graph, and let us follow a similar procedure
to~\autoref{algomaxconnected}. We begin by leaving the first and last nodes untouched and 
perform measurements on the leaf nodes of every one of the alternating $(m - 2) / 2$ switches. 
This leads to a discrepancy when we reach the final measurement step: both $Sw_{m-2}$ and 
$Sw_{m-1}$ retain their leaf nodes: $Sw_{m - 1}$ due to the initial hypothesis and $Sw_{m - 2}$
because of the alternating measurement pattern. So, at this stage, we must additionally 
remove the leaf nodes of $Sw_{m-2}$, which results in two consecutive leafless switches
by graph complementation. Consequently, the complete graph connectivity is reduced, since 
we must perform a $\sigma_Y$ measurement on $Sw_{m - 2}$ to proceed with the standard 
algorithm. In summary, when dealing with networks having an even number of switches, we 
must first transform it into an odd-numbered network of switches in order to be able to  
apply the proposed algorithm. Since we cannot assume that introducing additional switches 
with leaf nodes is allowed, we conclude that the only possibility is to remove at least 
one switch, thereby strictly reducing the network size by one star.

\begin{figure}[t]
    \centering
    \includegraphics[width=1\columnwidth]{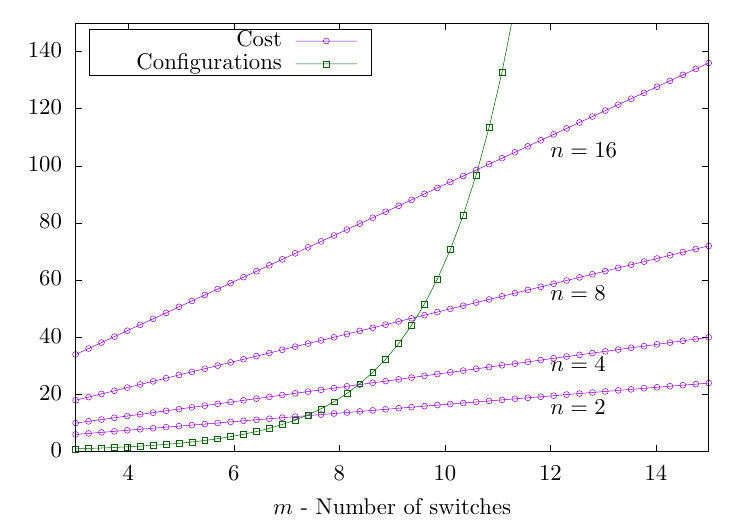}
    \caption{\label{fig:GraphStatesCost} Cost of the graph state transformation vs. number of switches $m$ and nodes per switch.}
\end{figure}

\subsection{Different number of leaf nodes in each switch}

We now aim to calculate the maximum allowed connectivity for a more realistic linear 
network where each switch has a different number of leaf nodes. Suppose that the 
network consists of $m$ switches, with each switch connected to $n_0, n_1, \dots, 
n_{m-1}$  leaf nodes, respectively. Let us define $T = \{n_0, \dots, n_{m-1}\}$ 
as the total set of leaf nodes. The total number of nodes in the network is given 
by $N_T = m + \sum_{i=0}^{m-1} n_i$. To determine the maximum number of neighbors
remaining, we calculate the total number of nodes and subtract the number of 
measurements performed. since each measurement removes a node from the network.  
We consider the case where $m$ is odd and proceed again according to the steps 
outlined in~\autoref{algomaxconnected}. First, select $\frac{m-1}{2}$ switches 
to perform $\sigma_Z$ operations onto. As discussed previously, these switches 
are chosen so as to remove all leaf nodes from every other switch. Thus, the set 
of selected nodes is $I = \{n_1, n_3, \dots, n_{m-2} \}$, and the total number 
of measurements made during this step is $M_1 = \sum_{i \in I} n_i$. Next, we 
perform $\frac{m - 1}{2}$ $\sigma_X$ operations.  This step also contributes a 
total of $M_2 = \frac{m-1}{2}$, and results in the star configuration
$\ket{S_{\Tilde{n_1}}}$ with $\Tilde{n_1} = \frac{m-1}{2} + \sum_{i \in T 
\setminus I} n_i$. Therefore, the total number of measurements implemented on 
the network is given by $M_T = M_1 + M_2 = \frac{m-1}{2} + \sum_{i \in I} n_i$,
and subtracting this from the total number of nodes, we obtain the remaining 
neighbors as $N_T - M_T = \frac{m-1}{2} + \sum_{i \in T \setminus I} n_i$,
where $T\setminus I = \{n_{2i}\}_{i=0}^{\frac{m-1}{2}}$. Thus, the new maximum 
neighbor connectivity for the original $\alpha$ is given by $\Tilde{\alpha} = 
\frac{m+1}{2} + \sum_{i \in T \backslash I} n_i$. Finally, it is straightforward 
to verify that, when $n_i \rightarrow n \quad \forall i \in T, \quad \Tilde{\alpha} 
\rightarrow \alpha$, and the complete interconnection graph $K_n$ is again realized.

\section{Other forms of Neighboring }
\label{subsec:other:forms:Neighboring}

Although \autoref{algomaxconnected} allows us to achieve maximum neighbor connectivity 
in the form of a complete graph resembling a GHZ state, it ultimately leads to a fixed 
outcome. The procedure is constrained to a particular set of gates on specific 
nodes and does not generalize beyond these conditions. However, the results
in~\cite{Chen2024} introduce a degree of flexibility in peer-to-peer connectivity, 
provided  that a bi-star logical topology is first established. This Section is 
dedicated to achieving precisely that for our specific network configuration.

In order to obtain a bi-star configuration from an initial $m$-star configuration, 
we must perform a minimum of  $\frac{m-1}{2}$ measurements ---excluding the 
$\sigma_Z$ measurements on the leaf nodes. This requirement becomes evident 
when considering the underlying structure of the algorithm. In the final step, we 
obtain a system consisting of  two switches and their corresponding leaf nodes 
---one switch retaining the original leaves,  while the other accumulates additional 
leaf nodes—connected via an intermediary switch  without any leaf nodes. Applying a 
$\sigma_X$ gate at this stage results in a star configuration that is locally 
equivalent to a complete graph. However, if instead a $\sigma_Y$ is applied, the 
system transitions into a bi-star configuration.

While this particular case is not immediately compelling (since the resulting nodes 
maintain the same communication capabilities as in the original algorithm) 
selecting different switches for leaf node deletion allows for a broader range 
of communication possibilities. The challenge arises naturally from this intuition: 
we must identify  $\frac{m-1}{2}$ nodes whose neighbors will be removed, followed by 
the application of the appropriate Pauli gates to establish a bi-star topology. The 
complexity of this task escalates rapidly, as it evolves into a combinatorial 
optimization problem: assuming an odd number $m$ of switches, we would have 
$m-2$ switches from which to choose, and we must select $\frac{m-1}{2}$ of them. 
Hence, the number of possible configurations is  $\binom{m-2}{\frac{m-1}{2}} \approx 
2^{m/2}$ for large $m$, which becomes computationally unmanageable even for a 
relatively small number of switches  (e.g., $m=9$).  The problem is more tractable 
if we exploit  certain properties of the network topology. Since the network has 
an odd number of switches,  it is always possible to identify a central switch 
located at position $\frac{m-1}{2}$. This central symmetry allows us to simplify the 
problem by treating several configurations as equivalent.  Specifically,  
we can use the symmetry $i \rightarrow m- 1 - i$, $i = 0, \dots, \frac{m-1}{2}$,
around this central node as identical configurations. 

To illustrate this, we consider the example with $m = 7$ (the case $m = 5$ is trivial, 
as it results  in only two distinct straightforward configurations). The number of 
distinct cases is $\binom{5}{3}= 10$, and the following symmetries hold $1 
\leftrightarrow 5$, $2 \leftrightarrow 4$, $3$ unchanged. Using these facts and 
exploring the few possibilities, we arrive at a total of $6$ distinct cases, 
enumerated for the sake of illustration in~\autoref{tab:client_removal}. 
There are three distinct topological results, 
in terms of maximizing the number of neighboring connections, that arise when 
performing sequential gates on the network. The first result $Sw_1, Sw_3, Sw_5$ 
yields the same configuration as the one presented in \autoref{algomaxconnected} 
and has already been studied. Another interesting configuration is the tri-star, 
a structure consisting of three star graphs joined at their central nodes -- Tri-
star($4^\dag, 5, 6$)  Tri-star($0, 5^\dag, 6$) . In this topology, one node, marked 
with $\dag$, holds the majority of the clients, while the other two nodes retain 
their original connections. Although this sequence of gates does not maximize the 
number of neighboring connections, it enables different sets of allowed communications
between nodes. A similar scenario occurs for $Sw_1, Sw_2, Sw_4$ arriving to the
$Bi-star(0,6^\dag)$ topology.  An intriguing case arises in  the form 
$Bi-star(0^\dag, 6^\dag)$, where both switches hold a similar number  of clients. 
If we were to follow standard protocols of removing all but one client from one 
of the nodes, we would compromise our goal of achieving maximum connectivity, 
thereby  negating the intended purpose of this result. Instead, we adopt an 
alternative protocol in~\cite{Chen2024} to the realization of an extranet 
artificial topology, in our interpretation a bi-star topology interconnecting 
each client of one star  with each client of the other. Precisely, through the 
sequence of gates $\sigma_X(0,6) \rightarrow \sigma_Z(6)$, where nodes $6$ 
and $0$ are interchangeable, yielding the same configuration $\ket{K_{n_1,n_2}}$, 
where $K_{n_1, n_2}$ denotes the complete bipartite graph with $(n_1, n_2)$ in each 
class.  ~\autoref{fig:GraphStatesCost} depicts the cost
(number of measurements) of the proposed transformations. 

\begin{table}[t]
\centering
\caption{Clients removed and corresponding transformations}
\label{tab:client_removal}
\begin{tabular}{|c|cc|c|} \hline
\textbf{$Sw$} & \multicolumn{1}{c|}{\textbf{First Step}} & \textbf{Second Step} & \textbf{Result} \\ \hline
 $\{ 1, 3, 5 \}$ & \multicolumn{2}{c|}{\autoref{algomaxconnected}} & Bi-star ($0, 6^\dag$) \\ \hline
 $\{ 1, 2, 3 \}$  & \multicolumn{1}{c|}{$\sigma_X(1,0)$} & $\sigma_X(3,2)$ & Tri-star ($4^\dag, 5, 6$) \\ \hline
 $\{ 2, 3, 4 \}$ & \multicolumn{1}{c|}{$\sigma_X(2,1)$} & $\sigma_X(4,3)$ & Tri-star ($0, 5^\dag, 6$) \\ \hline
 $\{ 1, 2, 4 \}$ & \multicolumn{1}{c|}{$\sigma_Y(1)$} & $\sigma_X(4,5)$, $\sigma_X(2,3)$ & Bi-star ($0, 6^\dag$) \\ \hline
 $\{ 1, 2, 5 \}$ & \multicolumn{1}{c|}{$\sigma_Y(1)$} & $\sigma_X(5,4)$, $\sigma_X(2,3)$ & Bi-star ($0^\dag, 6^\dag$) \\ \hline
 $\{1, 3, 4 \}$ & \multicolumn{1}{c|}{$\sigma_Y(3)$} & $\sigma_X(4,5)$, $\sigma_X(1,2)$ & Bi-star ($0, 6^\dag$) \\ \hline
\end{tabular}
\end{table}

\section{Discussion and Application}
\label{sec:discussion}

\subsection{Quantum Networks} 

Though the Quantum Internet has not yet been fully conceptualized, we adhere to 
the consensus that it would operate under similar principles as the classical Internet,
supported by a suite of Quantum Internet protocols to enable seamless communication 
between devices. Thus, a Global Quantum Network will be composed of 
interconnected quantum sub-networks. Our proposal formalizes hierarchical networks
composed of bi-star and tri-star topologies, i.e., a Quantum Internet backbone 
plus Quantum Islands as a possible network-wide  architecture.

This structure has already been demonstrated in QKD deployment
(Euro QCI~\cite{garciacid2024strategiesintegrationquantumnetworks}), where 
each QKD domain has its own SDN controller responsible for intra-domain services, 
and these controllers communicate via an East-West Bound Interface to support 
inter-domain communication. Based upon this idea, we propose a hybrid model 
combining bipartite quantum networks and graph-state networks. In this model, 
intra-domain connectivity relies on bipartite entanglement distribution, ensuring 
better compatibility with classical communication and control, while graph states 
would enable the generation and manipulation of complex entangled states within 
each domain, aligning with the network topology. An architectural design featuring 
a two-layer network enables cluster reconfiguration, see~\cite{Clayton2024}. 
Such techniques allow arbitrary networks to be transformed into structured
configurations, such as multi-star networks, facilitating the efficient execution 
of the algorithms presented in this article.

\subsection{Efficient Quantum Protocols}
\label{aplications:protocols}

Graph states provide the necessary entanglement structure for secure and 
coordinated multi-party quantum protocols. 

\subsubsection{Multi-user QKD} 
One first  application is multi-user QKD for an arbitrary subset \( S \) of 
users. Instead of constructing entanglement through bipartite links between all 
user pairs—requiring \( \binom{|S|}{2} \)  connections ---a more efficient 
approach involves generating a GHZ state for all clients and discarding those not 
in \( S \). While not optimal, this method eliminates the need for Steiner trees 
and significantly reduces resource  overhead~\cite{chelluri2025multipartiteentanglementdistributionbellpair}.

\subsubsection{Quantum Conference Key Agreement (QCKA)} This is a widely studied 
field that aims to extend Quantum Key Distribution (QKD) to multiple
parties~\cite{Murta_2020}. Consider a total of $n$ users who wish to establish a 
shared secret key: Alice and the set of users $Bob_1, \dots, Bob_{n-1}$. This can 
be achieved using an $n$-partite GHZ state, which represents the final quantum 
state obtained in our protocols. A set of stabilizer measurements can be performed to 
detect potential eavesdroppers by establishing parity checks. Additionally, 
performing $\sigma_Z$ or $\sigma_X$ measurements on each node enables the full key 
to be extracted securely. Such configurations are essential for overcoming the 
limitations of bipartite entanglement, enabling multi-party Quantum Key Distribution
(QKD)~\cite{Qian2012} and quantum secret sharing~\cite{Markham2008}.

\subsubsection{Distributed Quantum Computing} Another important application of 
graph states is found in grid topologies, where vertices 
are positioned at the intersections of a grid. This structure is particularly relevant for
Distributed Quantum Computing (DQC)~\cite{Mao2024} and the widely used $X$ protocol, which
enables bipartite entanglement between any two nodes in the grid. This protocol is further
optimized in~\cite{Negrin2024}, where improvements in resource efficiency and robustness 
against noise in dynamic and lossy networks are demonstrated.

\section{Conclusions}
\label{sec:conclusions}

In this work, we have extended and refined the methods proposed in~\cite{deJong2024} 
and \cite{Chen2024}, thereby enhancing their applicability and generality. 
Building upon these, we expanded their algorithm to a linear cluster
architecture, where client nodes are attached to each of the central switches, 
arriving at a fixed maximum GHZ state for a multi-star network, depending on the 
number of switches $m$ and their neighbors $n$. We worked out in full detail
a particular network case for $m=7$ switches and an arbitrary number of client nodes.

A key open problem is to derive a general result for node deletion, 
leveraging network symmetries and node selection strategies to mitigate the 
exponential growth in complexity. Further exploration in this direction could lead 
to significant optimizations in large-scale quantum network architectures. A wider
area for future research lies in the area of Distributed Quantum Computing (DQC);  
we point at exploring mesh and grid  topologies to extract similar results to those 
of this paper. 




 
\appendix
\label{appendix: apendixa}


Since the graph is bicolorable (coloring scheme shown in \autoref{fig:multistar}), we define two subsets of nodes and the edge set as $E = E_S \cup E_{n_i}$ where $E_S$ represents the sequential connections between switches, while $E_{n_i}$
represents the connections between each switch and its  leaf nodes 
$K_j^i$. Thus, the graph is given by $G = (P_0,P_1, E)$.
\begin{align*}
    P_0 &= \bigl\{ Sw_{2i} \bigr\}_{i=0}^{\frac{m-1}{2}} \cup \bigl\{ K_j^{2i+1} \bigr\}_{j=1,i=0}^{n, \frac{m-3}{2}}, \\
    P_1 &= \bigl\{ Sw_{2i+1} \bigr\}_{i=0}^{\frac{m-3}{2}} \cup \bigl\{ K_j^{2i} \bigr\}_{j=1,i=0}^{n, \frac{m-1}{2}} \\
    E_S = & \bigl\{ (Sw_i,Sw_{i+1}) \bigr\}_{i=0}^{m-2}, \quad E_{n_i} = \bigl\{ (Sw_i, K_j^i) \bigr\}_{j=1,i=0}^{n, m-1}. 
\end{align*}

Following  \autoref{algomaxconnected}, we start by selecting alternating 
nodes, beginning with $Sw_0$ (they align with the odd-labeled switches). Then, apply $\sigma_Z$ measurements on their leaf nodes, modifying 
$P_0$ while leaving $P_1$ unchanged. As a result, the edge set is updated to
\begin{align*}
    P_0' = &\bigl\{ Sw_{2i} \bigr\}_{i=0}^{\frac{m-1}{2}}, \quad P_1' = P_1.\\
       E_S' = &E_S, \quad E_{n_i}' = \bigl\{ (Sw_{2i}, K_j^{2i}) \bigr\}_{j=1,i=0}^{n,\frac{m-1}{2}},
\end{align*}
where $E' = E_S \cup E'_{n_i}$. The modified graph is therefore $G'=(P_0',P_1,E')$.
Next, we perform a $\sigma_X$ operation on every odd switch, where the $\sigma_X$ 
operator acts as $\sigma_X (G)_{\alpha, \beta} = \tau_{\beta}(\tau_{\alpha}(\tau_{\beta}(G))-\alpha))$.
To illustrate this process, consider a minimal cell affected by a $\sigma_X$ 
operation when applied at $\alpha=Sw_{2i+1}$ with $\beta=Sw_{2i}$. The relevant
node subset is
\begin{align*}
    \Tilde{V} = & \bigl\{ Sw_{2i+1}, Sw_{2i}, Sw_{2i-1}, Sw_{2i+2}, K_j^{2i}, K_j^{2i+2} \bigr\}_{j=1}^n \\
    \Tilde{E} =& \bigl\{ (Sw_l, Sw_{l+1}) \bigr\}_{l=2i-1}^{2i+1}  \cup
    \bigl\{ (Sw_h, K_j^h) \bigr\}_{j=1}^n
\end{align*}
where $h=\{2i,2i+2\}$Local complementation at vertex $Sw_{2i}$ results in a complete graph on the subset
$V = \{Sw_{2i-1}, Sw_{2i+1}, Sw_{2i}, K_j^{2i}\}_{j=1}^n$. Since no neighboring 
nodes of $Sw_{2i}$ were previously connected, new edges appear
\begin{align*}
    \Tilde{E}' = & \{(K_j^{2i}, K_h^{2i})\}_{j \neq h} \cup \{(Sw_{2i-1}, Sw_{2i+1})\} \cup \\
    & \{(Sw_{2i-1}, K_j^{2i})\}_{j=1}^n \cup \{(Sw_{2i+1}, K_j^{2i})\}_{j=1}^n.
\end{align*}
Additionally, $Sw_{2i+2}$ becomes connected to $Sw_{2i+1}$ along with its leaf nodes $K_j^{2i+2}$. Performing local complementation at $Sw_{2i+1}$ removes the edges within the previously complete subgraph, yielding the updated edge set
\begin{align*}
    \Tilde{E}'' = & (Sw_{2i+1}, \{Sw_{2i}, Sw_{2i-1}, Sw_{2i+2}, K_j^{2i}\})_{j=1}^n \cup \\
    & (Sw_{2i+2}, \{Sw_{2i}, Sw_{2i-1}, K_j^{2i}, K_j^{2i+2}\})_{j=1}^n.
\end{align*}
Since the next step is the deletion of $Sw_{2i+1}$, all its connections are removed, 
leading to the new graph $\Tilde{V}''' = \{Sw_{2i-1}, Sw_{2i}, Sw_{2i+2}, K_j^{2i}, K_j^{2i+2}\}_{j=1}^n$, $\Tilde{E}''' = Sw_{2i+2} \times\{Sw_{2i}, Sw_{2i-1}, K_j^{2i}, K_j^{2i+2}\}_{j=1}^n$.
This structure corresponds to a star graph centered at $Sw_{2i+2}$.
Finally, since $Sw_{2i}$ has only one neighbor, the local complementation operation at 
$Sw_{2i}$ is trivial, leaving the graph unchanged.
By iterating this process throughout the network, we find that the final output of the 
algorithm is a larger star graph consisting of all even-indexed switches and their 
respective leaf nodes
$Sw_{2i}, K_j^{2i}, \quad \forall i\in \{0, \frac{m-1}{2}\}, \quad j\in \{1,n\}$.
This graph is locally equivalent to a complete graph of the same size, which can be achieved by performing a local complementation at its central vertex. \qed

\end{document}